\documentclass[a4paper,11pt]{article}
\usepackage{amsmath, amssymb}
\usepackage{bbm}
\usepackage{caption}
\usepackage{subcaption}
\usepackage{graphicx}
\usepackage{color}
\usepackage{graphicx}
\usepackage{url}
\usepackage{color,soul}
\usepackage{float}
\usepackage{amsmath}
\usepackage{amssymb}
\usepackage{courier}
\usepackage{multicol}
\usepackage[abs]{overpic}
\usepackage{enumitem}
\usepackage{multicol}
\usepackage{amsthm}
\newtheorem{theorem}{Theorem}
\newtheorem{lemma}[theorem]{Lemma}
\newtheorem{claim}[theorem]{Claim}

\usepackage[linesnumbered,boxed]{algorithm2e}
\allowdisplaybreaks
\usepackage{pgf,tikz}
\usepackage{mathrsfs}
\usetikzlibrary{arrows}
\usepackage{lineno}

\def\RR{\mathbb R}
\def\QQ{\mathbb Q}
\def\ZZ{\mathbb Z}

\def\cC{\mathcal C}

\def\cF{\mathcal F}

\def\cL{\mathcal L}

\def\cR{\mathcal R}

\def\cV{\mathcal V}

\def\b1{\mathbf 1}

\def\eps{\varepsilon}

\def\sees{\sim}
\newcommand{\raf}[1]{(\ref{#1})}

\newcommand{\poly}{\operatorname{poly}}
\newcommand{\polylog}{\operatorname{polylog}}

\newcommand{\argmax}{\operatorname{argmax}}
\newcommand{\conv}{\operatorname{conv.hull}}
\newcommand{\area}{\operatorname{area}}
\newcommand{\prem}{\operatorname{prem}}

\newcommand{\cells}{\operatorname{\cC}}

\newcommand{\Vis}{\operatorname{\cV}}

\newcommand{\bone}{\ensuremath{\boldsymbol{1}}}
\newcommand{\hide}[1]{}

\def \VCd {\text{VC-dim}}

\def \MO {\textsc{Max}}

\def \SO {\textsc{Subsys}}

\def \OPT {\textsc{Opt}}

\newtheorem{observation}{Observation}

\newtheorem{remark}{Remark}

\usepackage[top=1in, bottom=1in, left=1in, right=1in]{geometry}

\title{A Deterministic Bicriteria Approximation Algorithm for  the Art Gallery Problem
}
\author{
	Khaled Elbassioni\thanks{Khalifa University of Science and Technology, P.O. Box 127788, Abu Dhabi, UAE;
		(khaled.elbassioni@ku.ac.ae)}
}

\date{}
\begin{document}
	
\maketitle
%\linenumbers

\begin{abstract}
 Given a polygon $H$ in the plane, the art gallery problem calls for fining the smallest set of points in  $H$ from which every other point in $H$ is seen. We give a deterministic algorithm that, given any polygon $H$ with $h$ holes, $n$ rational veritces of maximum bit-length $L$, and a parameter $\delta \in(0,1)$, is guaranteed to find a set of points in $H$ of size $O\big(\OPT\cdot\log(h+2)\cdot\log (\OPT\cdot\log(h+2)))$ that sees at least a  $(1-\delta)$-fraction of the area of the polygon. The running time of the algorithm is polynomial in $n$, $L$ and $\log(\frac{1}{\delta})$ and quasi-polynomial in $h$, where $\OPT$ is the size of an optimum solution.
\end{abstract}

\section{Introduction}

In the art gallery problem, we are given a polygon $H$ with $n$ vertices and  $h$ holes, and two sets of points $N$, $G$ in $H$. 
Two points $p,q$ in $H$ are said to {\it see} each other if the line segment joining them lies inside $H$. The objective is to {\it guard} all the points in $N$ using candidate {\it guards} from $G$, that is, to 
find a subset $G'\subseteq G$ such that for every point $q\in N$, there is a point $p\in G'$ such that $p$ sees $q$.

This problem has received a considerable attention in the literature, see, e.g., \cite{BCKO08,O87} and the references therein. If one of the sets $G$ or $N$ is an (explicitly given) discrete set, then the problem can be easily reduced to a standard \textsc{SetCover} problem. For the case when $G=V$ is the vertex set of the polygon (called {\it vertex} guards), Ghosh \cite{G87, G10} gave an $O(\log n)$-approximation algorithm that runs in time $O(n^4)$ for simple polygons (resp., in time $O(n^5)$ for non-simple polygons). In fact, a direct application of the $\epsilon$-net theory, and the fact that the range space $\cF_H$ defined by the visibility polygons of points in $H$ has bounded VC-dimension, implies that the vertex guarding problem for simple polygons admits an $O(\log\OPT)$-approximation algorithm, where $\OPT$ here is the size of an optimum set of {\it vertex} guards.  This result was improved by King \cite{K13} to $O(\log\log\OPT)$-approximation in time $O(n^3)$ for simple polygons (resp.,  $O(\log h\cdot\log\OPT)$-approximation in time $O(h^2n^3)$ for non-simple polygons). The main ingredient for the improvement in the approximation ratio is the fact proved by King and Kirkpatrick \cite{KK11} that there is an $\epsilon$-net, in this case and in fact more generally when $G=\partial H$ (called {\it perimeter guards}), of size $O(\frac{1}{\epsilon}\log\log \frac{1}{\epsilon})$.

For the  (infinite) case when $N=G=H$ (called, sometimes, the {\it point} or {\it unrestricted} guarding problem), a {\it discretization} step, which selects a candidate discrete set of guards guaranteed to contain a near optimal-set, seems to be necessary for reducing the problem to \textsc{SetCover}. Such a discretization method was given in \cite{DKDS07} and was claimed to give an $O(\log\OPT)$-approximation for simple polygons (resp., $O(\log h\cdot\log(\OPT\cdot\log h))$-approximation in non-simple polygons) in {\it pseudo-polynomial} time $\poly(n,\Delta)$, where $\OPT$ is the size of an optimal solution, and the {\it stretch} $\Delta$ is defined as the ratio between the longest and shortest distances between vertices of $H$. However, later on, an error in one of the claims in \cite{DKDS07} was pointed out by Bonnet and Miltzow \cite{BM16},  who also suggested another discretization procedure that results in an $O(\log\OPT)$-randomized approximation algorithm, after making the following two assumptions:
\begin{itemize}
	\item[(A1)] vertices of the polygon have integer components, given by their binary representation; 
	\item[(A2)] no three extensions meet in a point of $H$ that is not a vertex, and no three vertices are collinear; here an extension is a line passing through two vertices.
\end{itemize}	                                                                                                                      
Under these assumptions, it was shown in \cite{BM16} that one can use a grid $\Gamma$ of cell size $\frac{1}{D^{O(1)}}$ such that $\OPT_\Gamma=O(\OPT)$, where $\OPT_{\Gamma}$ is the size of an optimum set of guards {\it restricted} to $\Gamma$, and $D$ is the diameter of the polygon. Then one can use the algorithm suggested by Efrat and Har-Peled \cite{EH06} who gave an efficient implementation of the {\it multiplicative weights update} (MWU) method of Br\"{o}nnimann and Goodrich \cite{BG95} in the case when the locations of potential guards are restricted to a dense grid $\Gamma$. More precisely, the authors in \cite{EH06} gave a {\it randomized} $O(\log\OPT_{\Gamma})$-approximation algorithm for simple polygons (resp., $O(\log h\cdot\log(\OPT_{\Gamma}\cdot\log h))$-approximation in non-simple polygons) in expected time $O(n\OPT_{\Gamma}^2\cdot\log\OPT_{\Gamma}\cdot\log(n\OPT_{\Gamma})\log^2 \Delta)$ (resp.,  $nh\OPT_{\Gamma}^3\polylog n\cdot\log^2 \Delta)$), where $\Delta$ here denotes the ratio between the diameter of the polygon and the grid cell size. Note that this would imply a randomized (weakly) polynomial-time polylogarithmic factor-approximation algorithm for the unrestricted guarding case, if one can show that for a rational polygon (i.e., one with vertices having rational coordinates), there is a near optimal set which has also a rational description (and hence, the bit length representation of the output guarding set can be bounded by the bit length of the input). It is well-known that, if a rational polygon $H$ can be guarded with one guard, then there is an optimal guard placement with rational coordinates. Interestingly, Abrahamsen, Adamaszek and Miltzow~\cite{AAM17} showed that there is a rational polygon that can be guarded with three guards, but requires four guards if the guards are required to have rational coordinates. More recently, Meijer and Miltzow~\cite{MM22} strengthened this result by giving an example of a rational polygon, where an optimal solution of size two must to be irrational, thus settling the question of the minimum number of guards required to have irrational coordinates. The latter result implies that any approximation algorithm for the art gallery problem that runs in polynomial time in the {\it bit model of computation} cannot have an approximation factor better than $\frac32.$  While it is not clear whether this lower bound can be substantially improved, the authors in \cite{BM16} showed that the ratio $\Delta$ above can be chosen, under assumption (A2), to be $D^{O(1)}$, which implies by (A1) that $\log \Delta$ is linear in the maximum bit-length of a vertex coordinate. In particular, under assumptions (A1) and (A2), there is a polylogarthmic rational approximation for the optimal guarding set.

%Another interesting question is regarding the bit length representation of the output guarding set. 

%Note that, the same argument combined with Theorem~\ref{BG} in the appendix shows that one can actually obtain $O(\log z_{\cF}^*)$-approximation in randomized polynomial-time for simple polygons under assumptions (AG1) and (AG2). 
%Furthermore, we will show below that Theorem~\ref{t-main2} implies that there is a deterministic algorithm for achieving the same ratio, under the  same assumptions.

On the hardness side, the vertex (and point) guarding problem for simple polygons is known to be APX-hard~\cite{ESW01}, while the problem for non-simple polygons is as hard as \textsc{SetCover} and hence cannot be approximated by a polynomial-time algorithm with ratio less than %$((1-\epsilon)/12)\log n)$, for any $\epsilon>0$, unless $\text{\it NP}\subseteq\text{\it TIME}(n^{O(\log\log n)})$.  
$(1-\epsilon)\ln n$, for any $\epsilon>0$, unless $\text{\it P}=\text{\it NP}$ \cite{DS14}.
More recently, it was shown by Abrahamsen, Adamaszek and Miltzow~\cite{AAM22} that the point guarding problem is $\exists\RR$-complete, that is, the problem is equivalent under polynomial time reductions to deciding whether a system of polynomial equations with integer coefficients over the real numbers has a solution.

\paragraph{Our contribution.}~Given that there does not yet exist a polylogarithmic approximation algorithm for (the general variant of) the art gallery problem, we investigate in this paper a more relaxed question: given $\delta\in(1,0)$, can we obtain an $\alpha:=\polylog(h,\OPT)$-approximation of the optimal set by a permitting a $\delta$-fraction of the polygon to be non-guarded? Such approximation guarantees are typically called $(\alpha,1-\delta)$-{\it bicriteria} approximation, and can be, for all practical purposes, as good as normal approximations, provided that $\delta$ can be made sufficiently small. While it is straightforward to get such an approximation algorithm with running time $\poly(\frac1\delta)$ (see Appendix~\ref{rand-samp}), the question arises  whether the running time can be made polylogarithmic in $\frac1\delta$.  In~\cite{E17,E23}, we gave an affirmative answer to this question, in fact for any range space of bounded VC-dimension (under some mild assumptions), by presenting a {\it randomized} $(\alpha,1-\delta)$-bicriteria approximation algorithm whose running time is polylogarithmic in $\frac1\delta$. In this paper, we strengthen this result by giving a {\it deterministic}  algorithm with the same guarantees for the art gallery problem. 
\begin{theorem}\label{t1}
	Given a polygon $H$ with  $h$ holes and $n$ vertices with rational representation of maximum bit-length $L$ and $\delta>0$, there is a deterministic algorithm that finds in $\poly(L,n,$ $h^{\log\log (h\cdot\OPT)},\log\frac{1}{\delta})$ time a set of points in $H$ of size $O(\OPT\cdot\log(h+2)\cdot\log (\OPT\cdot\log(h+2)))$ and bit complexity $\poly(L,n,h,\log\frac{1}{\delta})$, guarding at least $(1-\delta)$-fraction of the area of $H$, where $\OPT$ is the value of the optimal solution. 
\end{theorem}

As most of the previous results~\cite{AES12,AP20,BM17,EH06}, we prove Theorem~\ref{t1} using a variant of the MWU method~\cite{AHK06,BG95}, which can be applied generally for finding an approximate minimum hitting set for a range space $\cF=(Q,\cR)$ given by its set of points $Q$ and set of ranges $\cR\subseteq 2^Q$. When the set $Q$ is finite, the MWU method can be implemented by maintaining weights on the points of $Q$, which are multiplicatively updated from one iteration to the next. For range spaces $\cF$ of bounded VC-dimension, this approach yields an $O(\log\OPT)$-approximation algorithm in time polynomial in $|Q|$.  For the art gallery problem, the range space is defined by the visibility polygons of the points in $H$ which implies that set of points $Q$ is infinite. The straightforward extension of the approach in~\cite{BG95} to this infinite (or continuous) case does not seem to work, since the bound on the number of iterations depends on the area of the regions created during the course of the algorithm, which can be arbitrarily small (see Appendix~\ref{sec:BG} for more details). %Another variant of MWUs that goes via zero-sum games was considered by Agarwal and Pan in~\cite{AP20} where weights are maintained on both points and ranges. When trying to extend this approach to infinite range spaces, we face the same difficulty w.r.t. bounding the running time as the one mentioned above. 
Nevertheless, it turns out that a simple modification of this method can be made to work. We highlight the main ingredients below.
%In this paper, we take a third approach, which has been standardly used in the optimization literature (see, e.g., \cite{AHK06,GK07}), but does not seem to be much (if at all) considered in the computational geometry literature: 
\begin{itemize}
	
	\item We follow~\cite{ERS05,PA95} in obtaining first a {\it fractional} solution $\widehat \mu$ using a variant of MWU method which resembles (in essence) the one used in~\cite{GK07} for approximately solving {\it covering linear programs}. At a high level, the MWU algorithm (implicitly) maintains weights on the points of $H$, initially all set to $1$. In each iteration, it finds (approximately) a point $p$ whose visibility polygon has maximum total weight, then it decreases the weights of all points seen by $p$. As in~\cite{GK07}, to be able to guarantee that the algorithm terminates in polynomial time, a small twist is added to this (standard) procedure. Namely, when the weight of a point has been decreased ''enough'', it is (implicitly) removed before  calling the maximization procedure (typically called {\it maximization oracle}). Once the total area of all removed points is at least $(1-\delta)$-fraction of the area of $H$, the MWU algorithm terminates. The set $\widehat P$ of points returned by the sequence of calls to the maximization oracle defines the support of the fractional solution $\widehat\mu$. 
	
	\item Given the fractional solution $\widehat\mu:\widehat P\to\RR_+$, and the fact that the range space $\cF_H$ has bounded VC-dimension, we use the known deterministic methods for computing $\epsilon$-nets~\cite{BCM99,CM96,M91} to get a guarding set satisfying the bounds in Theorem~\ref{t1}.
	
	\item ({\bf Main new ingredient}) The success of the above approach hinges on being able to implement the maximization oracle called in each iteration of the MWU algorithm. Under the bit model of computation, we show that such a maximization oracle can be implemented in deterministic polynomial time. The idea, following \cite{NT94}, is to express the function representing the total weight of the visibility region of a point $q\in H$ as a sum of continuous functions, each of which is a ratio of two polynomials in two variables, namely, the $x$ and $y$-coordinates of $q$. Then maximizing over $q$ amounts to solving a system of two polynomial equations of degree $\poly(n,h,\log\frac{1}{\delta})$ in two variables, which an be solved using {\em quantifier elimination techniques}, e.g., \cite{BPR96,GV88,R92}. Since we need to apply this procedure in each iteration of the MWU algorithm, a technical hurdle that we have to overcome is that the required bit length may grow from one iteration to the next, resulting in an exponential blow-up in the bit length needed for the computation.  To deal with this issue, we need to round the arrangement of lines arising in the course of the algorithm in in each iteration so that the total bit length in all iterations remains bounded by a polynomial in the input size.  
    
\end{itemize}

%It should be noted that the main new ingredient in comparison with~\cite{E17,E23} is the deterministic implementation of the maximization oracle. 
\medskip

\paragraph{Comparison with other approaches.}~In contrast, we mention the following standard approaches (see Appendix~\ref{alt} for details):

\begin{itemize}
	\item Standard {\it greedy} approach: we keep adding a point that sees the (approximately) largest area of the non-guarded region, until only a $\delta$-fraction of the area remains unguarded. This can be easily shown (see~\cite{CEH07} and also Section~\ref{greedy}) to have an approximation ratio of $O(\log\frac1\delta)$, which is much worse than the bound given by our theorem when $\delta $ is sufficiently small. Interestingly, our approach, which can be considered as a generalization of the greedy strategy (combined with MWUs) allows to 'move' the $\log\frac1\delta$ factor from the approximation ratio to the running time.
%	\item Random sampling: we obtain a random sample $Q\subseteq H$ of size $r:=O(\frac{\OPT\cdot\log\OPT}{\delta^2}\log\frac{\OPT}{\sigma})$, and then solve a discrete hitting set problem on $H$, where the requirement is to only guard all points in $Q$ (using the  approach in~\cite{BG95}). It can be easily seen that this gives the same approximation guarantees as in our main theorem, but with running time depending {\it polynomially} on $\frac{1}{\delta}$; see Section~\ref{rand-samp}. 

\item Direct Application of the Br\"{o}nnimann-Goodrich Algorithm\cite{BG95}: it can be shown (see~\cite{E16}) that the algorithm can be extended to work for the (full coverage version of) the art gallery problem, but the running time may depend on the area of the smallest cell $R$ of the arrangement created in the course of the algorithm. Since $R$ can be arbitrarily small, there is no guarantee that such an extension will terminate in polynomial (ore even finite) time.    
\end{itemize}
 
\medskip
 
The rest of the paper is organized as follows. In the next section we define our notation and recall some preliminaries. In Section~\ref{sec:algorithm},  for completeness,  we describe the MWU algorithm from~\cite{E23} and give a slightly simplified version of its analysis written in the context of the art gallery problem. The main result of the paper, which is the deterministic implementation of the maximization oracle, is given in Section~\ref{sec:det}. Details of some the alternative approaches (mentioned above) for obtaining a bicriteria approximation algorithm for the art gallery problem are  given briefly in Appendix~\ref{alt}.

\section{Preliminaries}\label{sec:prelim} 
\subsection{Notation}
We regard a polygon $H$ as a set of points in $\RR^2$ (consisting of all the points inside $H$ including the boundary $\partial H$). We assume that $H$ has $n\ge 3$ vertices and $h\ge 0$ holes and denote the set of vertices of $H$ by $V(H)$. % and assume we are given two sets $N,G\subseteq H$, called demand and guard sets, respectively. 
Two points $p,q\in H$ are said to {\it see} each other, denoted by $p\sees q$, if the line segment joining them is a subset of $H$. 
For a point $q\in H$, we denote by $\Vis(q)=\{p\in H~:~p\sees q\}$ the visibility polygon of $q$ with respect to (w.r.t.) $H$; see Fig.~\ref{fig12}.   
%Equivalently, $p\sees q$ if and only if $p\in \Vis_H(q)$ (or equivalently, $q\in \Vis_H(p)$). 
For a set of points $Q\subseteq H$, we also write $\Vis_Q(q):=\{p\in Q:~q\sees p\}$ be the set of points in $Q$ that see $q$.

An instance of the art gallery problem is given by a polygon $H$ of $n$ vertices and $h$ holes and a coverage parameter $\delta\in(0,1)$. A feasible solution to the given instance is defined by a set of points $P\subseteq H$ such that every point in $H$ is seen from some point in $P$, that is, $H=\cup_{p\in P}\Vis(p)$.  An optimal solution is a feasible solution that minimizes $|P|$, and its value is denoted by $\OPT$. A $(1-\delta)$-feasible solution is defined by a set of points $P\subseteq H$ such that $\area(\cup_{p\in P}\Vis(p))\ge(1-\delta)\area(H)$.
For $\alpha\ge 1$, a bicriteria $(\alpha,1-\delta)$-approximate solution is a $(1-\delta)$-feasible solution $P$ such that $|P|\le\alpha\cdot\OPT$.
An $(\alpha,1-\delta)$-approximate {\it fractional} solution $(P,\mu)$ is defined by a finite set $P\subseteq H$ and a nonnegative function $\mu:P\to\RR_+$ such that $\mu(\Vis(q)):=\sum_{p\in P\cap \Vis(q)}\mu(p)\ge 1$ for al $q\in H'$ and $\mu(P)\le\alpha\cdot\OPT$, for some $H'\subseteq H$ having $\area(H')\ge(1-\delta)\area(H)$.

\begin{figure}	
	\centering  
	\begin{subfigure}{.4\textwidth}
		\includegraphics[width=2.5in]{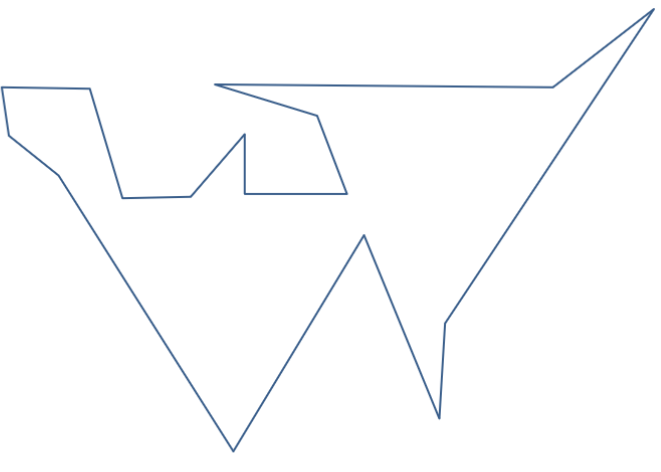}
		\caption{A simple polygon $H$.}   
		\vspace{0.15in}                                          \label{fig1}
	\end{subfigure}
	\hspace{0.18in}
	\begin{subfigure}{.4\textwidth}
		\includegraphics[width=2.5in]{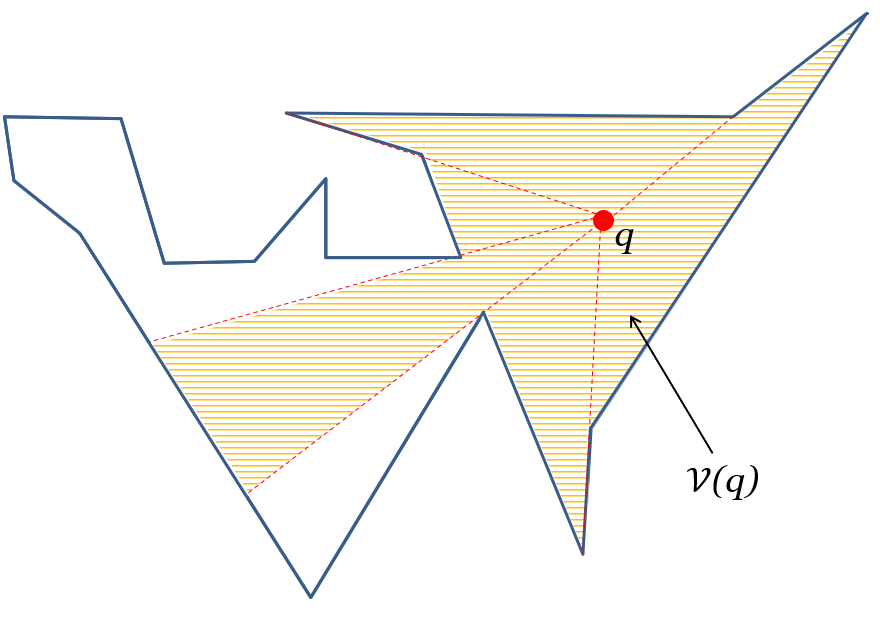}
		\caption{The visibility polygon of a point $q \in H$.} 
		\label{fig2}
	\end{subfigure}
\caption{Visibility regions in a polygon.}\label{fig12}
\end{figure}

\subsection{Range Spaces of Bounded VC-Dimension and $\epsilon$-Nets}\label{sec:VCd}
Given a range space $\cF=(Q,\cR)$ defined by a set of points $Q$ and a set of ranges $\cR\subseteq 2^Q$, its  {\it VC-dimension} is defined as follows. A finite set $P\subseteq Q$ is said to be {\it shattered} by $\cF$ if $\cR_{|P}:=\{R\cap P~:~R\in\cR\}=2^P$. The VC-dimension of $\cF$, denoted $\VCd(\cF)$, is the cardinality of the largest subset of $Q$ shattered by $\cF$. If arbitrarily large subsets of $Q$ can be shattered then $\VCd(\cF)=+\infty$. Given a polygon $H$, we can define a range space $\cF_H=(Q,\cR)$, where $Q=H$ and $\cR=\{\Vis(q)~:~q\in H\}$. Valtr \cite{V98} showed that $\VCd(\cF_H)\le 23$ for any simple polygon $H$ (i.e., a polygon without holes) and that $\VCd(\cF_H)=O(\log h)$ for any polygon $H$ with $h$ holes. For simple polygons, this has been improved to $\VCd(\cF_H)\le 14$ by Gilbers and Klein \cite{GK14}.                                                         

Given a range space $\cF=(Q,\cR)$, a finite set $P\subseteq Q$, a nonnegative (weight) function $\mu:P\to\RR_+$, and a parameter $\epsilon\in(0,1)$, an {\it $\epsilon$-net} for $\cR$ (w.r.t. $\mu$) is a set $P'\subseteq P$ such that $P'\cap R\neq\emptyset$ for all ranges $R\in\cR$ satisfying $\mu(R)\ge\epsilon\cdot\mu(P)$ (where for $Q'\subseteq Q$, we write $\mu(Q'):=\sum_{q\in P\cap Q'}\mu(q)$). 
We say that a range space $\cF$ admits an $\epsilon$-net of size $s_{\cF}(\cdot)$, if for any $\epsilon\in(0,1)$, there is an $\epsilon$-net of size $s_{\cF}(\frac{1}{\epsilon})$. For range spaces of VC-dimension $d$, it is well-known \cite{HW87,KPW92} that a random sample w.r.t. the probability measure $\frac{\mu}{\mu(P)}$ of size $s_{\cF}(\frac{1}{\epsilon})=O(\frac{d}{\epsilon}\log\frac{1}{\epsilon})$ is an $\epsilon$-net with (high) probability $\Omega(1)$. 
It is also known~\cite{BCM99,CM96,M91} that an $\epsilon$-net for $\cR$ of size $s_{\cF}(\frac{1}{\epsilon})=O(\frac{d}{\epsilon}\log\frac{d}{\epsilon})$
can be computed deterministically in time $O(d)^{3d}\frac{1}{\epsilon^{2d}}\log^d(\frac{d}{\epsilon})|P|$, under the assumption that 
the range space is given by a {\it subsystem oracle} \SO$(\cF,Q')$ that, given any finite $Q'\subseteq Q$, returns the set of ranges $\cR_{|Q'}$ in time $O(|Q'|)^{d+1}$.  It is easy to see that, for the range space $\cF_H$ defined by a polygon $H$ and  a finite set $Q'\subseteq H$, such a subsystem oracle \SO$(\cF,Q')$  can be implemented in $O(n(h+1)|Q'|^2\log(n(h+1)|Q'|))$ time  by computing the arrangement defined by the visibility polygons of the points in $Q'$; see~\cite{EH06} and the next section for more details.

\subsection{The MWU Algorithm}   \label{sec:algorithm}
The algorithm, shown as Algorithm \ref{alg} below, proceeds in iterations $t=0,1,2,\ldots$. At any iteration $t$, the algorithm maintains a (possibly multi-) set of candidate guard locations $P_t\subset H$ (initially $P_0=\emptyset$), which is extended by exactly one point in each iteration. 
Let us define the {\it active} polygon $H_t\subseteq H$ as
\begin{align}
H_t := \{q \in H~:~|P_t\cap \Vis(q)| < T\}.\label{Ht}
\end{align}
%For convenience, we assume below that $P_t$ is  a multi-set (repetitions allowed).

In the algorithm , we assume that there is an {\it oracle} \MO$(H,H',w,\nu)$ that, given a polygon $H$, a set of points $H'\subseteq H$, a weight function $w:H\to\RR_+$, and an accuracy parameter $\nu\in(0,1)$, returns a point $p\in H$ such that
$$
w(\Vis_{H'}(p))\geq (1-\nu)\max_{q\in H}w(\Vis_{H'}(q)).
$$
(In that sense, the algorithm can be considered as an iterative variant of a weighted greedy algorithm.)
We will show in Section~\raf{sec:det} how such an oracle can be implemented in deterministic polynomial time.

%{\algomargin}{.25in}
\begin{algorithm}[H]
	\label{alg}
	\SetAlgoLined
	\KwData{A polygon $H$ and approximation accuracies $\eps,\delta,\nu\in(0,1)$}
	\KwResult{A $(\frac{1+2\eps}{1-\nu},1-\delta)$-approximate solution $P\subseteq H$}
	
	$t\gets 0$; $P_0\gets\emptyset$; $T\gets\frac1{\eps^2}\ln\frac{1}{\delta}$; $H_0\gets H$\\
			$w_0(q)\gets 1$ for all $q\in H_0$\\
	\While{$\area(H_t)\ge\delta\cdot\area(H)$}{
		$p_{t+1}\gets\MO(H,H_t,w_t,\nu)$ \label{s-oracle}\\
		$P_{t+1}\gets P_{t}\cup\{p_{t+1}\}$\\
		\For{$q\in H$	}{
			\tcp{Implicitly}	    
			$w_{t+1}(q) \gets \left\{\begin{array}{ll}(1-\eps)w_t(q)&\text{ if }q\in \Vis(p_{t+1})\cap H\\
	    w_t(q)&\text{ otherwise }\end{array}\right.$ \label{s-wt-red}}
		$t \gets t+1$\\
	}
    $\widehat P\gets P_t$; $\widehat \mu(p)\gets\frac1T$ for all $p\in\widehat P$; $\epsilon\gets\frac1{\widehat\mu(\widehat P)}$\\
	\Return an $\epsilon$-net of size $s_{\cF_H}(\frac1\epsilon)$ from the set $P_t$
	\caption{The MWU (weighted greedy) algorithm.}
\end{algorithm}

Given the set of points $P_t$ in iteration $t$, we get a partition of $H$ into a set of regions $\cR_t$ such that for any $R\in \cR_t$ all points in $R$ see the the same subset of points from $P_t$, that is, for all $p,q\in R$, $\Vis_{P_t}(p)=\Vis_{P_t}(q)$; see Fig.~\ref{fig34} for an illustration. Such a partition can be found in polynomial time in $n$, $h$ and $|P_t|=t$.  Indeed, it can be easily seen that the set of (visibility polygons of) the points in $P_t$, together with the set of vertices $V(H)$ of the polygon $H$, induces an arrangement of line segments (in $\RR^2$) of total complexity (say, number of faces) $r_t:=O(n(h+1)t^2)$; see., e.g.,~\cite{EH06}.  Thus, we can construct the set of (disjoint convex) cells of this arrangement, call it $\cC_t(H)$, in $O(n(h+1)t^2\log(n(h+1)t))$ time\footnote{For simplicity of presentation, we do not attempt here to optimize the running time, for instance, by maintaining a data structure for computing $\cells_t(H)$, which can be efficiently updated when a new point is added to $P_t$.}, and label each cell $R\in\cC_t(H)$ of the arrangement by the set of points $P_t(R)$ from $P_t$ it sees (e.g., using a sweep line algorithm), that is, $P_t(R)=\{p\in P_t~:~p\sees q\text{ for all }q\in R\}$. 

\begin{figure}
	\centering  
	\begin{subfigure}{.4\textwidth}
		\includegraphics[width=2.5in]{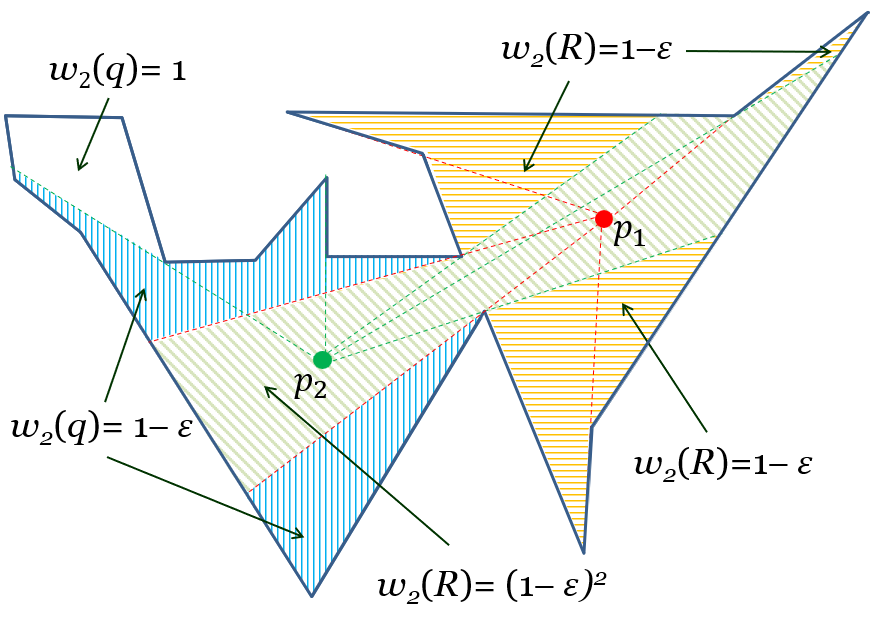}
		\caption{The partition $\cR_2$ of the polygon $H$ defined by the set of points $P_2=\{p_1,p_2\}$. Each (color) pattern defines one region of the partition. The weight function $w_2$ is also shown, where all points in one region have the same weight. For $T=2$, $H_2$ does not include the diagonally  striped region.}                                             \label{fig3}
	\end{subfigure}
	\hspace{0.18in}
	\begin{subfigure}{.4\textwidth}
		\includegraphics[width=2.5in]{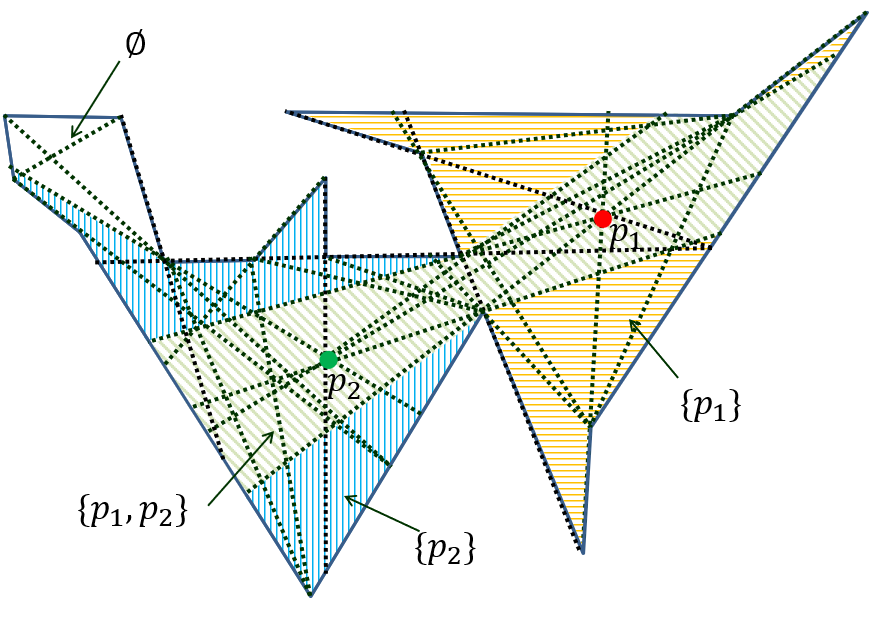}
		\caption{The set $\cC_2(H)$  of convex cells defined by the arrangement of lines induced by the set of points in the union of the vertices of $H$ and $P_2$. The sets $P_2(R)$ are shown for some cells $R\in\cC_2(H)$.} 
			\vspace{0.33in}
		\label{fig4}
	\end{subfigure}
\caption{The partition of $H$ defined by the set of points $P_t$  for $t=2$.}	\label{fig34}
\end{figure}

\subsection{Analysis}\label{sec:analysis}
We first make a few easy observations that will be useful in the analysis of the algorithm.

\begin{observation}\label{ob1}
	For all $t\ge 0$, $R\in\cells_t(H)$, it holds that $w_t(R)=(1-\eps)^{|P_t(R)|}w_0(R)$. Consequently, $w_t(H_t)>(1-\eps)^Tw_0(H_t)$.
\end{observation}
\begin{proof}
	The first claim can be easily seen by induction on $t\ge0$ with the base case being trivial. Assume the claim holds for some $t>0$, and consider $R\in\cells_{t+1}(H)$. Note that $\cells_{t+1}(H)$ is a refinement of  $\cells_t(H)$ as some of the sets in $\cells_t(H)$ may be partially seen by the new point $p_{t+1}$. If $R\not\subseteq \Vis(p_{t+1})$, then $w_{t+1}(R)=w_t(R)=(1-\eps)^{|P_t(R)|}w_0(R)=(1-\eps)^{|P_{t+1}(R)|}w_0(R)$; otherwise, line~\ref{s-wt-red} of the algorithm implies that $w_{t+1}(R)=(1-\eps)w_t(R)=(1-\eps)^{|P_{t}(R)|+1}w_0(R)=(1-\eps)^{|P_{t+1}(R)|}w_0(R)$.
	
	The second claim is immediate from the first claim as def.~\raf{Ht} of $H_t$ implies
	\begin{align*}
	w_t(H_t)&=w_t\left(\bigcup_{R\in\cells_t(H):~|P_t(R)|<T}R\right)= \sum_{R\in \cells_t(H):~|P_t(R)|<T}w_t(R)\\&=\sum_{R\in\cells_t(H):~|P_t(R)|<T}(1-\eps)^{|P_t(R)|}w_0(R)\\&>(1-\eps)^T\sum_{R\in\cells_t(H) :~|P_t(R)|<T}w_0(R)=(1-\eps)^Tw_0(H_t).
	\end{align*}
\end{proof}

%\begin{observation}\label{ob4}
%	For all $t\ge 0$, $w_t(H_t)>(1-\eps)^Tw_0(H_t)$.
%\end{observation}
%\begin{proof}
%	Since $|P_t(R)|<T$ for all $R\in\cells_t(H)$ such that $R\subseteq H_t$, we get by  Observation~\ref{ob1} that 
%	$w_t(H_t)>(1-\eps)^T\sum_{R\in\cells_t(H)}w_0(R)=(1-\eps)^Tw_0(H_t)$.
%\end{proof}

The analysis of the algorithm follows from the following three lemmas. The proofs are essentially the same as in~\cite{E23} (with some simplifications as~\cite{E23} compares the algorithm's output to an optimum {\it fractional} solution) and are included for completeness.
\iffalse
\begin{lemma} \label{l1} For all $t=0,1,\ldots$, it holds that
	\begin{align}
	w_{t+1}(N_{t+1}) \leq w_t(N_t) \exp\left(-\frac{\eps}{w_t(N_t)}\cdot w_t(\Vis_{N_t}(p_{t+1}))\right). \label{eq3}
	\end{align}
\end{lemma}

\begin{proof}
	As $N_{t+1}\subseteq N_t$,
	\begin{align*}
	w_{t+1}(N_{t+1})&\le w_{t+1}(N_t)
	=  w_t(N_{t})-\eps\cdot w_t(\Vis_{N_t}(p_{t+1}))\nonumber \\
	&=  w_t(N_t)  \left(1-\eps\frac{ w_t(\Vis_{N_t}(p_{t+1}))}{w_t(N_t)}\right) 
	\leq  w_t(N_t)\exp\left(-\eps\frac{ w_t(\Vis_{N_t}(p_{t+1}))}{w_t(N_t)}\right).
	\end{align*}
	%where the last inequality follows from the inequality $1-z \leq e^{-z}$, valid all $z$.
\end{proof}

\begin{lemma}
	For all $t=0,1,\ldots$, it holds that $\OPT\cdot\frac{w_{t}(\Vis_{N_t}(p_{t+1}))}{w_t(N_t)} \geq 1-\nu$. \label{lem1}
\end{lemma}

\begin{proof}
	Consequently, for a $(1+\eps)$-approximate solution $\mu^*$, Observations~\ref{ob2} and~\ref{ob3}, Lemma~\ref{l-0}, and~\raf{eq2} imply that the following holds with probability at least $1-\sigma$.
	\begin{align*}
	\OPT\cdot\frac{w_{t}(\Vis_{N_t}(p_{t+1})}{w_t(N_t)} 
	&\ge\OPT\cdot(1-\nu)\frac{\max_{q\in H}w_t(\Vis_{N_t}(q))}{w_t(N_t)}\\
	&\ge \frac{1-\nu}{1+\eps}  \cdot\frac{w_{t}(\cR_{t})}{w_t(\cR_t)}
	= \frac{1-\nu}{1+\eps},
	\end{align*}
	where the first inequality is due to the $(1+\eps)$-approximability of $\mu^*$, the second inequality is due to (\ref{eq2}), and the last inequality is due to the feasibility of $\mu^*$ for \raf{FH} and Lemma~\ref{l-0}.
	%, and the last equality is because the increments in all iterations are unity.
\end{proof}
\fi
\begin{lemma} \label{l1} For all $t=0,1,\ldots$, it holds that
	\begin{align}
	w_t(H_t) \leq w_0(H) \exp\left( -\frac{\eps(1-\nu)t}{\OPT}\right).
	\end{align}
\end{lemma}
\begin{proof}
	We first prove that
		\begin{align}\label{e1-1}
	w_{t+1}(H_{t+1})<\exp\left(-\frac{\eps(1-\nu)}{\OPT}\right)w_{t}(H_t).
	\end{align}
	Consider an optimal solution $P^*\subseteq H$. Then (as $P^*$ sees all points in $H\supseteq H_t$), 
	\begin{align}\label{e1-1-1-1}
	w_t(H_t)&=w_t\left(\bigcup_{q\in P^*}\Vis_{H_t}(q)\right)\le \sum_{q\in P^*}w_t(\Vis_{H_t}(q)).
	\end{align}
	From \raf{e1-1-1-1} it follows that there is a point $q\in P^*$ such that  $w_t(\Vis_{H_t}(q))\ge\frac{w_t(H_t)}{\OPT},$ 
	and thus by the choice of $p_{t+1}$, $w_t(\Vis_{H_t}(p_{t+1}))\ge(1-\nu)w_t(\Vis_{H_t}(q))\ge\frac{(1-\nu)w_t(H_t)}{\OPT}$. Consequently, as $H_{t+1}\subseteq H_t$, we get
	\begin{align*}%\label{e1-1-3}
	w_{t+1}(H_{t+1})&\le w_{t+1}(H_{t})=w_t(H_t)-\eps \cdot w_t(\Vis_{H_t}(p_{t+1}))\\
	&\le w_t(H_t)-\eps\cdot\frac{(1-\nu)w_t(H_t)}{\OPT}=\left(1-\frac{\eps(1-\nu)}{\OPT}\right)w_t(H_t)\\&<\exp\left(-\frac{\eps(1-\nu)}{\OPT}\right)w_t(H_t).
	\end{align*}
	This gives~\raf{e1-1}, repeated applications of which yields the claim.
\end{proof} 

Define $t_{\max}:=\frac{\OPT}{\eps(1-\nu)}\left(T\ln\frac{1}{1-\eps}+\ln\frac{1}{\delta}\right)$.

\begin{lemma}\label{l-bd1} After at most $t_{\max}$ iterations, we have $w_0(H_{t_{\max}})<\delta\cdot w_0(H)$.
\end{lemma}
\begin{proof}
	Lemma~\ref{l1} implies that, for $t:=t_{\max}$, we have
	\begin{align*}
	w_t(H_t)<e^{-\frac{\eps(1-\nu)t}{\OPT}}w_{0}(H).
	\end{align*}
	By Observation~\ref{ob1}, $w_t(H_t)>(1-\eps)^Tw_0(H_t)$. Thus, if $w_0(H_t)\ge\delta\cdot w_0(H)$, we get
	\begin{align*}
	(1-\eps)^T\delta<e^{-\frac{\eps(1-\nu)t}{\OPT}},
	\end{align*}
	giving $t<\frac{\OPT}{\eps(1-\nu)}\left(T\ln\frac{1}{1-\eps}+\ln\frac{1}{\delta}\right)=t_{\max}$, in contradiction to $t=t_{\max}$.
\end{proof}

\begin{lemma}\label{l-bd2} Using $T=\frac1{\eps^2}\ln\frac{1}{\delta}$ and $\eps \leq 0.68$, the while-loop in Algorithm~\ref{alg} terminates in at most $t_{\max}$ iterations with a $(\frac{1+2\eps}{1-\nu},1-\delta)$-approximate fractional solution $(\widehat P,\widehat\mu)$ for $H$.
\end{lemma}
\begin{proof}
	By Lemma~\ref{l-bd1},  the while-loop in Algorithm \ref{alg} terminates in iteration $t_f \le t_{\max}$.
	We define the fractional solution $(\widehat P,\widehat\mu)$ of support $\widehat P=P_{t_f}$ and with $\widehat\mu(p)=\frac1T$ for $p\in \widehat P$.
	\medskip
	
	\noindent{\it $(1-\delta)$-Feasibility:} By the stopping criterion, $w_0(H_{t_f})<\delta\cdot w_0(H)$. Moreover, for any $q\in H\setminus H_{t_f}$, we have
	$\widehat\mu(\Vis(q))=\sum_{p\in \widehat P\cap \Vis(q)}\widehat\mu(p)=\frac{1}{T}|\widehat P\cap \Vis(q)|\ge 1$, since $|\widehat P\cap \Vis(q)|\ge T$, for all $q\in H\setminus H_t$.
	
	\medskip
	
	\noindent{\it Quality of the fractional guarding set $\widehat \mu$:}
	Using $T=\frac1{\eps^2}\ln\frac{1}{\delta}$, we have
	\begin{align*}
	\frac{\widehat\mu(\widehat P)}{\OPT}&=\frac{|\widehat P|}{T\cdot \OPT}=\frac{t_f}{T\cdot z^*_\cF}\le\frac{t_{\max}}{T\cdot \OPT}=\frac{1}{\eps(1-\nu)}\left(\ln\frac{1}{1-\eps}+\frac1T\ln\frac{1}{\delta}\right)\\
	&\le\frac{1}{1-\nu}\left(\frac{1}{\eps}\cdot\ln\frac{1}{1-\eps}+\eps\right)<\frac{1+2\eps}{1-\nu},	
	\end{align*}
	for $\eps \leq 0.68$. 
\end{proof}

    Note that the running time is {\it output-sensitive} as it depends on the value of $\OPT\le n$.

    \medskip
    
\noindent{\it Proof of Theorem~\ref{t1}.}
Given the $(\frac{1+2\eps}{1-\nu},1-\delta)$-approximate fractional solution $(\widehat P,\widehat\mu)$ for $H$, obtained at the end of the while-loop in Algorithm~\ref{alg} (as guaranteed by Lemma~\ref{l-bd2}), we use a similar argument as in \cite{ERS05} for rounding the fractional solution. Let $\epsilon:=\frac1{\widehat\mu(\widehat P)}$. By definition of $\widehat \mu$, we have $\widehat\mu(\widehat P)\le\frac{1+2\eps}{1-\nu}\cdot\OPT$, and $\widehat\mu(\Vis(q))\ge 1=\epsilon\cdot\widehat\mu(\widehat P)$ for all $q\in H'$, where $H'\subseteq H$ has $\area(H')\ge(1-\delta)\area(H)$. If $\epsilon=1$ than $1=\widehat\mu(\widehat P)\ge\widehat \mu(\Vis(q))\ge 1$, for all $q\in H'$, implies that $\widehat P\subseteq\Vis(q)$ for all $q\in H'$, and thus taking any point in $\widehat P$ gives a solution satisfying the claim in the theorem. Thus we may assume that $\epsilon<1$.
Consider the range space $\cF_{H'}=(H,\{\Vis(q)~:~q\in H'\})$ defined by the visibility polygons of points in $H'$ (w.r.t. $H$), together with the weight function $\widehat\mu:\widehat P\to\RR_+$. According to the results stated in Section~\ref{sec:VCd}, $\VCd(\cF_{H'})\le\VCd(\cF_H)=O(\log (h+2))$ and an $\epsilon$-net $P\subseteq \widehat P$ of size 
\begin{align*}
|P|&=s_{\cF_H}\left(\frac{1}{\epsilon}\right)=O(\widehat\mu(\widehat P)\cdot\log(h+2)\cdot\log(\widehat\mu(\widehat P)\cdot\log( h+2)))\\&=O\left(\frac{1+2\eps}{1-\nu}\cdot\log(h+2)\cdot\log\Big(\frac{1+2\eps}{1-\nu}\cdot\OPT\cdot\log(h+2)\Big)\right)\cdot\OPT
\end{align*} 
can be computed deterministically in time $\poly(n,h^{\log\log (h\cdot|\widehat P|)},|\widehat P|)$.  The definition of $P$ implies that $P\cap \Vis(q)\neq\emptyset$ for all $q\in H'$. The theorem follows.  
\qed

\section{A deterministic Maximization Oracle} \label{sec:det}

In this section, we show how the oracle \MO$(H,H',w',\nu)$ can be implemented in polynomial time, for a given a polygon $H$  and an accuracy parameter $\nu\in[0,1)$, when $H'=H_t\subseteq H$ is the active polygon in iteration $t$ of Algorithm~\ref{alg}, and  $w'=w_t:H\to\RR_+$ is the weight function in iteration $t$.
We assume that the components of the vertices of $H$ have rational representation, with maximum bit-length $L$ for each component.  

\subsection{Oracle Implementation}
In a given iteration $t$ of Algorithm~\ref{alg}, we are given an active polygon $H_t\subseteq H$, with $\area(H_t)\ge\delta\cdot\area(H)$, determined by the current set of chosen points $P_t\subseteq H$, and the current weight function $w_t:H\to\RR_+$, given by $ w_t(p)=(1-\eps)^{|P_t\cap \Vis(p)|}$, for $p\in H$. Let $\cR_t$ be partition of $H$ defined by the set $\{\Vis(p):~p\in P_t\}$ of visibility polygons of points in $P_t$. As explained in Section~\ref{sec:algorithm}, the set of (convex) cells induced by $\cR_t$ over $H$ has complexity $r_t=O(n(h+1)|P_t|^2)$ and can be computed in time $O(n(h+1)|P_t|^2\log(n(h+1)|P_t|))$; as before, let us call this set $\cells_t(H)$, and for any $R\in\cells_t(H)$, define $P_t(R)=\{p\in P_t:~p\sees q\text{ for all } q\in R\}$ (recall that, for all $R\in\cells_t(H)$, all points in $R$ are equivalent w.r.t. visibility from $P_t$). Noting that $H_t=\bigcup_{R\in\cells_t(H):~|P_t(R)|<T}R$, we can write $\xi_t(q):=w_t(\Vis_{H_t}(q))$ for any $q\in H$ as
\begin{equation}\label{vis-wt}                                                                                                    
\xi_t(q)=\sum_{R\in \cells_t(H):~|P_t(R)|<T}(1-\eps)^{|P_t(R)|}\area(\Vis(q)\cap R).
\end{equation}
Now, to find a point $q$ in $H$ maximizing $\xi_t(q)$, we follow the approach in\footnote{Note that \cite{NT94}  only considers the case when $w_t\equiv 1$.} %\footnote{It should be noted that an FPTAS was claimed in \cite{NT94} when $w_t\equiv 1$, but this claim was not substantiated with a rigorous proof. In fact one of the statements leading to this claim does not seem to be correct, namely that the visibility region of the maximizer $q^*$ can be covered by a {\it constant} number of points that can be described only in terms of the input description of the polygon.} 
\cite{NT94} in expressing $\xi_t(q)$ as a (non-linear) continuous function of two variables, namely, the $x$ and $y$-coordinates of $q$. To do this, we first construct a refined partition of $H$ into a set of disjoint convex cells, induced by the arrangement of lines formed by the union $U_t$ of  $V(H)$ and the set of (internal) vertices of $\cells_t(H)\setminus V(H)$; see Fig.~\ref{fig5}. Let us denote this partition, which is a refinement of $\cells_t(H)$, by $\cells_t'(H)$ and note that it has complexity $O(r_t^2)=O(n^2(h+1)^2|P_t|^4)$. By definition, for any convex cell $Q$ in $\cells_t'(H)$, any two points in $Q$ are equivalent w.r.t. the visibility of points from $U_t$, that is, if $q,q'\in Q$, then $\cV_{U_t}(q)=\cV_{U_t}(q')$; see Fig.~\ref{fig6}. Moreover, for any pair of distinct vertices $p,p'$ of a cell $R\in\cells_t(H)$, any two points in $Q$ lie on the same side of the line through $p$ and $p'$ (otherwise, this line would have split $Q$ into at two different cells). This implies that, for any point $q=(x,y)\in Q$, the set $\Vis(q)\cap R$
can be decomposed into at most $|E|=O(r_t)$ polygons that are either convex quadrilaterals or triangles, where $E$ is the set of edges of $R$; see Fig.~\ref{fig67} for an illustration. Consider one such polygon $Z\subseteq R$, and assume without loss of generality (w.l.o.g.) that $Z$ is a quadrilateral. Using the notation in Fig.~\ref{fig7}, we can write the vertices of the quadrilateral $Z$ in counterclockwise order as $q_i=(\frac{a_i(x,y)}{b_i(x,y)},\frac{c_i(x,y)}{d_i(x,y)})$, for $i=1,\ldots,4$, where $a_i(x,y),b_i(x,y), c_i(x,y)$, and $d_i(x,y)$ are affine functions of the form $Ax+By+C$, for some constants $A,B,C\in\QQ$ which are multi-linear of degree at most $3$ in the components of some of the vertices of $R$ and $H$. By the {\it Shoelace formula}, we can further write the area of $Z$ as   
\begin{equation}\label{areaZ}
\area(Z)=\frac{1}{2}\sum_{i=1}^4\frac{a_i(x,y)}{b_i(x,y)}\left(\frac{c_{i+1}(x,y)}{d_{i+1}(x,y)}-\frac{c_{i-1}(x,y)}{d_{i-1}(x,y)}\right),
\end{equation}
where indices ($i+1$ and $i-1$) wrap-around from $1$ to $4$. By considering a triangulation of $Q_i$, and letting $\Delta$ be the triangle containing $q\in Q_i$, we can write $q$ in the parametric form $q=(x,y)=\lambda_1(x',y')+\lambda_2(x'',y'')+(1-\lambda_1-\lambda_2)(x''',y''')$, where $(x',y'),(x',y')$ and $(x''',y''')$ are the vertices of $\Delta$, and $\lambda_1,\lambda_2\in[0,1]$.

\begin{figure}
	\centering  
		\includegraphics[width=3.5in]{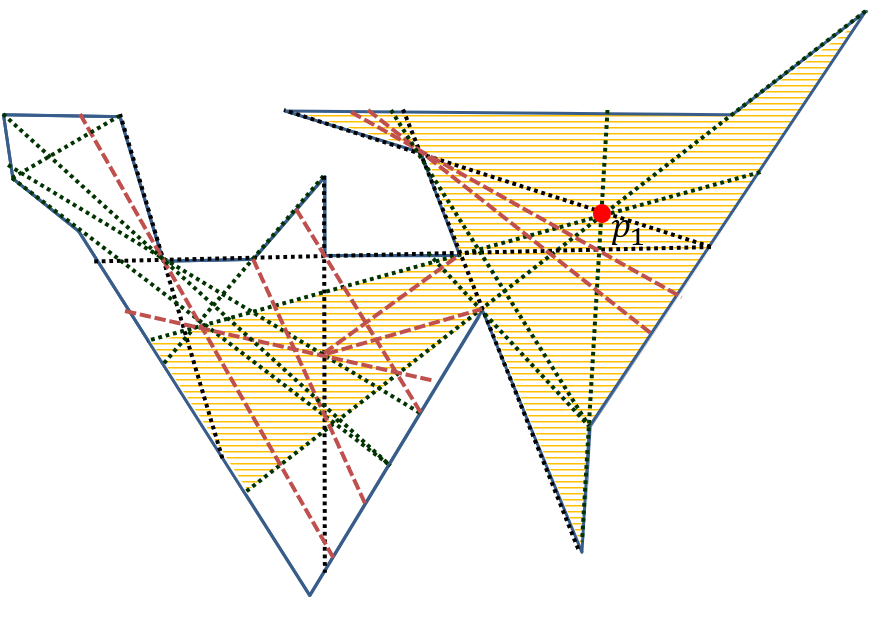}
		\caption{The refined partition $\cells'_1(H)$ obtained from $\cells_1(H)$ by adding all line segments between any pair of internal vertices of $\cells_1(H)$ and between any internal vertex of $\cells_1(H)$ and a vertex of $H$. For clarity of presentation, only some of the added segments are shown (as dashed lines).}                                             
		\label{fig5}
\end{figure}	

\begin{figure}
	\centering  
	\begin{subfigure}{.4\textwidth}
		\includegraphics[width=2.5in]{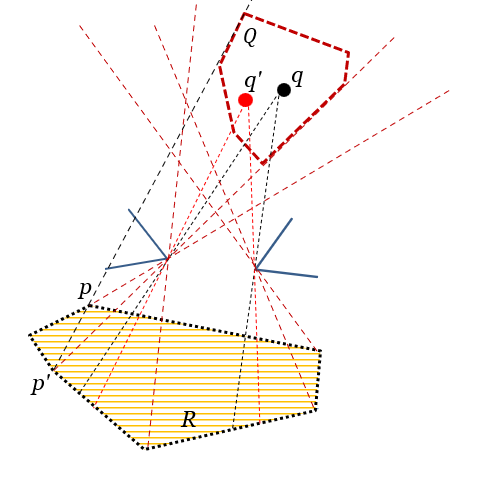}
		\caption{A cell $R\in\cells_t(H)$ and a cell $Q$ in the refinement $\cells_t'(H)$. Any two points $q,q'\in Q$ see exactly the same set of points from the set of vertices of $R$. The line through any two distinct vertices $p,p'$ of $R$ does not cut through $Q$. }                                             \label{fig6}
			\vspace{0.65in}
	\end{subfigure}
	\hspace{0.18in}
	\begin{subfigure}{.4\textwidth}
		\includegraphics[width=2.5in]{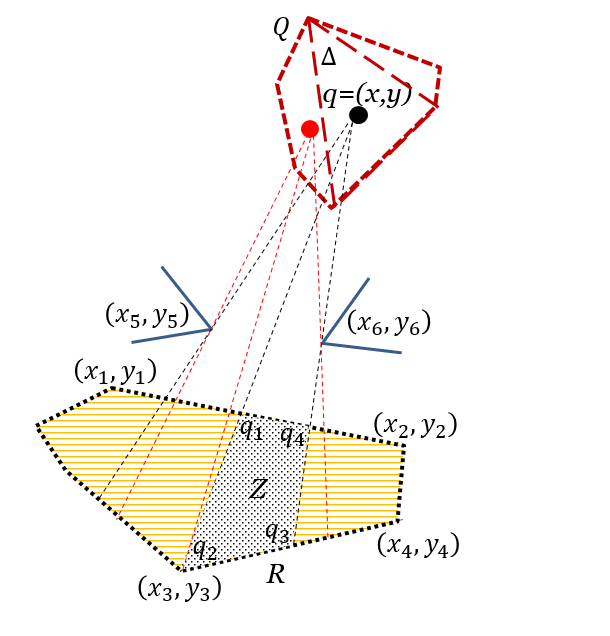}
		\caption{The visibility polygon of $q \in\Delta\subseteq Q$ inside $R$ can be decomposed into two convex quadrilaterals; one of them is $Z$, which is determined by $q=(x,y)$, the edges $\{(x_1,y_1),(x_2,y_2)\}$ and $\{(x_3,y_3),(x_4,y_4)\}$ of $R$, and the vertex $(x_6,y_6)$ of the polygon $H$.} 
		\vspace{0.54in}
		\label{fig7}
	\end{subfigure}
	\caption{The visibility polygon of a point $q\in Q$ inside $R\in\cells_t(H)$, where $Q\in\cells_t'(H)$.}	\label{fig67}
\end{figure}

It follows from \raf{vis-wt} and \raf{areaZ} that $\xi_t(q)$ can be written as     \begin{equation}\label{vis-wt2}                                                                                             \xi_t(q)=\xi_t(\lambda_1,\lambda_2)=\sum_{i=1}^k(1-\eps)^{j_i}\frac{M_i(\lambda_1,\lambda_2)}{N_i(\lambda_1,\lambda_2)},
\end{equation}     
where
$k=O(|E|)=O(r_t)=\poly(n,h,t)\le\poly(n,h,\log\frac{1}{\delta})$ (as $t\le t_{\max}=\poly(n,h,\log\frac{1}{\delta})$), $j_i\le T=O(\log\frac1\delta)$, and $M_i(\lambda_1,\lambda_2)$ and $N_i(\lambda_1,\lambda_2)$ are quadratic functions of $\lambda_1$ and $\lambda_2$ with coefficients having bit-length $O(L_t)$, where $L_t$ is the maximum bit length needed to represent the components of the vertices in $U_t$. We can maximize $\xi_t(\lambda_1,\lambda_2)$ over $\lambda_1,\lambda_2\in[0,1]$
by considering $9$ cases, corresponding to $\lambda_1\in\{0,1\}$, and $\lambda_1\in(0,1)$; and $\lambda_2\in\{0,1\}$, and $\lambda_2\in(0,1)$, and taking the value that maximizes $\xi_t(\lambda_1,\lambda_2)$ among them. Consider w.l.o.g. the case when $\lambda_1,\lambda_2\in(0,1)$. We can maximize $\xi_t(\lambda_1,\lambda_2)$ by setting the gradient of \raf{vis-wt2} to $0$, which in turn reduces to solving a system of two polynomial equations of degree $O(k)$ in two variables. A rational approximation to the solution $(\lambda_1^*,\lambda_2^*)$ of this system to within an {\it additive} accuracy of $\tau$ can be computed in time and bit complexity $\poly(k,L_t,\log\frac{1}{\tau})$, using, e.g., the {\em quantifier elimination algorithm} of Renegar \cite{R92}; see also Basu et al. \cite{BPR96} and Grigoriev and Vorobjov \cite{GV88}. 
\begin{claim}\label{cl4}
	The function $\xi_t(\lambda_1,\lambda_2)$ in \raf{vis-wt2} is $ 2^{C_0kL_t}$-Lipschitz, for some constant $C_0>0$.\footnote{A continuous differentiable function $f:S\to\RR$ is $\rho$-Lipschitz over $S\subseteq \RR^n$ if $|f(y)-f(x)|\le \rho\|x-y\|_2$ for all $x,y\in S$.}                                                                                                                                                                                                           
\end{claim}
\begin{proof}              
	It is enough to show that $\|\nabla\xi_t(\lambda_1,\lambda_2)\|_2\le 2^{O(kL_t)}.$ By \raf{vis-wt2}, each component of $\nabla\xi_t(\lambda_1,\lambda_2)$ is of the form $\frac{M(\lambda_1,\lambda_2)}{N(\lambda_1,\lambda_2)}$, where $M(\cdot,\cdot)$ and $N(\cdot,\cdot)$ are polynomials in $\lambda_1,\lambda_2\in[0,1]$ of degree $O(k)$ and coefficients of maximum bit length $O(kL_t)$. Thus $|M(\lambda_1,\lambda_2)|\le 2^{O(kL_t)}$. Also, from \raf{areaZ} and \raf{vis-wt2}, $N(\lambda_1,\lambda_2)$ can be written as a product of $k$ factors of the form $b(\lambda_1,\lambda_2)^2d(\lambda_1,\lambda_2)^2$, where $b(\cdot,\cdot)$ and $d(\cdot,\cdot)$ can be assumed to be {\it non-zero} affine functions of $\lambda_1$ and $\lambda_2$. Suppose $b(\lambda_1,\lambda_2)=A_1\lambda_1+A_2\lambda_2+A_0$, for some constants $A_0,A_1,A_2\in\QQ$ which have bit length $O(L_t)$. Since the minimum of $b(\lambda_1,\lambda_2)$ over $\lambda_1,\lambda_2\in[0,1]$ is attained at some $\lambda_1,\lambda_2\in\{0,1\}$, it follows that $|b(\lambda_1,\lambda_2)|\ge\min\{|A_0|,|A_1+A_0|,|A_2+A_0|,|A_1+A_2+A_0|\}\ge\frac{1}{2^{O(L_t)}}$. A similar observation can be made for $d(\cdot,\cdot)$ and implies that $N(\lambda_1,\lambda_2)\ge\frac{1}{2^{O(kL_t)}}$, which in turn implies the claim.       	
\end{proof}                                             
Let $q^*_\Delta\in\argmax_{q\in\Delta}\xi_t(q)$. If $\lambda^*:=(\lambda_1^*,\lambda_2^*)$ and  $\lambda:=(\lambda_1,\lambda_2)$ are the parametric descriptions of $q^*_\Delta$ and $q\in\Delta$, respectively, then one can easily verify that $\|q^*-q\|_2\le C_12^{L_t}\|\lambda^*-\lambda\|_2$ for some constant $C_1>0$.
Thus, by Claim~\ref{cl4}, we can choose the accuracy $\tau=\theta\cdot \frac{\sqrt{2}}{C_1}2^{-(C_0k +1)L_t}$ sufficiently small, to get a point $q_\Delta\in\Delta$ such that $\xi_t(q_\Delta)\ge \xi_t(q^*_\Delta)- \theta$, where  $\theta:=\frac{\nu\cdot w_t(H_t)}{n}\ge\frac{\nu \delta (1-\eps)^T\area(H)}{n}$ (by Observation~\ref{ob1} and the fact that $w_0(H_t)\ge \delta\cdot w_0(H)$ inside the while-loop of Algorithm~\ref{alg}), and hence $\log\frac{1}{\tau}=\poly(k,L_t,\log\frac{1}{\delta})$. Finally, we let $p_{t+1}\in\argmax_{\Delta} \xi_t(q_{\Delta})$, where $\Delta$ ranges over all triangles in the triangulations of the polygons $Q\in\cells_t'(H)$, to get
\begin{align*}
\xi_t(p_{t+1})=\max_{\Delta}\xi_t(q_\Delta)\geq\max_{\Delta}\xi_t(q^*_\Delta)-\theta=\max_{q\in H}\xi_t(q)-\frac{\nu\cdot w_t(H_t)}{n}\ge(1-\nu)\max_{q\in Q}\xi_t(q),
\end{align*}                                            
where the last inequality follows from the fact that $\max_{q\in H}\xi_t(q)\ge \frac{w_t(H_t)}{n}$.

\subsection{Rounding}
A technical hurdle in the above implementation of the maximization oracle is that the required bit length $L_t$ may grow from one iteration to the next (since the approximate maximizer $p_{t+1}$ above has bit length $\poly(n,h,L_t,\log\frac{1}{\delta})$), resulting in an exponential blow-up in the bit length needed for the entire computation. Indeed, given $p_{t+1}$, the new arrangement $\cells_{t+1}(H)$ is constructed by adding line segments connecting $p_{t+1}$ to some vertices in $H$, and hence some internal vertex in $\cells_{t+1}(H)$ may have bit length of $L_{t+1}=\poly(n,h,L_t,\log\frac{1}{\delta})$.

To deal with this issue, we need to round the polygonal arrangement $\cells_t(H)$ in each iteration (of the while-loop) so that the total bit length in all iterations remains bounded by a polynomial in the input size.\footnote{This is somewhat similar to the rounding step typically applied in numerical analysis to ensure that the intermediate numbers used during the  computation have finite precision.} This can be done as follows (see Fig.~\ref{fig89}). Recall that $\cells_t(H)$ consists of $r_t:=Cn(h+1)|P_t|^2$ disjoint convex polygons, for some constant $C>0$. Let $\bar t=\frac{n}{\eps(1-\nu)}\left(T\ln\frac{1}{1-\eps}+\ln\frac{1}{\delta}\right)$ be the upper bound on the number of iterations of the algorithm, and set 
$r_{\max}:= Cn(h+1)\bar t^2\ge \max_tr_t$. We consider an infinite grid $\Gamma$ in the plane of cell size $\rho=2^\ell$, where $\ell\in\ZZ$ is such that $\rho\le\frac{\nu\delta(1-\varepsilon)^T\area(H)}{16(1-\nu/2) D n r_{\max}}\le 2\rho$, and $D$ is the diameter of $H$ (which has bit length bounded by $O(L)$).

Let us call a cell $R\in\cells_t(H)$ {\it large} if $\area(R)\ge4\rho D$, and {\it small} otherwise. 
\begin{claim}\label{cl5}
	Every large cell  $R\in\cells_t(H)$ contains at least one grid point.
\end{claim}
\begin{proof}
Suppose that $R$ does not contain any point in $\Gamma$. 
By convexity, $R$ cannot overlap with more than one row and one column of $\Gamma$. As the area of any column or row in  $\Gamma$ inside $H$ does not exceed $\rho D$, we get the contradiction
$
2\rho D\ge \area(R)\ge 4\rho D.
$
\end{proof}

Let $\cL_t$ be the set of large cells in iteration $t$ of the algorithm. For each $R\in\cL_t$ we define an approximate polygon $\widetilde R\subseteq R$ as follows: for each vertex $v$ of $R$, we find a point $\widetilde v$ in $\Gamma\cap R$, closest to it (which is guaranteed to exist by Claim~\ref{cl5}), then define $\widetilde R:=\conv\{\widetilde v:~v \text{ is a vertex of $R$}\}$. Finally, we define the rounded polygonal arrangement $\widetilde{\cells}_t(H):=\{\widetilde R~:R\in\cells_t(H)\}$ and the rounded active polygon $\widetilde H_t:=\bigcup_{R\in\cL_t:~P_t(R)<T}\widetilde R$; see Fig.~\ref{fig8}.  

For $q\in H$, let us define the weight of the rounded visibility region $\widetilde\xi_t(q):=w_t(\Vis_{\widetilde H_t}(q))$:
\begin{equation}\label{vis-wt3}                                                                                                 
\widetilde\xi_t(q)=\sum_{R\in \cL_t:~|P_t(R)|<T}(1-\eps)^{|P_t(R)|}\area(\Vis(q)\cap\widetilde R),
\end{equation}
and note, by comparison to~\raf{vis-wt}, that 
\begin{align}\label{e3}
\xi_t(q)&=\widetilde\xi_t(q)+\sum_{R\in\cells_t(H)\setminus\cL_t:~|P_t(R)|<T}(1-\eps)^{|P_t(R)|}w_t(V(q)\cap R)\nonumber\nonumber\\
&\qquad~\quad+\sum_{R\in\cL_t:~|P_t(R)|<T}(1-\eps)^{|P_t(R)|}\area(\Vis(q)\cap (R\setminus\widetilde R))\nonumber\\
&\le\widetilde\xi_t(q)+\sum_{R\in\cL_t:~|P_t(R)|<T}\area(\Vis(q)\cap (R\setminus\widetilde R))\nonumber\\&=\widetilde\xi_t(q)+\area\Big(V(q)\cap\bigcup_{R\in\cL_t:~|P_t(R)|<T}(R\setminus\widetilde R)\Big)\le\widetilde\xi_t(q)+\area\Big(\Vis(q)\cap(H_t\setminus\widetilde H_t)\Big).
\end{align}
The only change we need in Algorithm~\ref{alg} is to replace $H_t$ in  by $\widetilde{H}_t$ and $\nu$ by $\frac\nu{2}$, when applying the maximization oracle in iteration $t$, that is, to replace line~\ref{s-oracle} of the algorithm by $$p_{t+1}\gets\MO(H,\widetilde H_t,w_t,\frac{\nu}{2}).$$ 
%It is easy to see that the analysis goes through with almost no change; we just have to replace $H_t$ by $\widetilde{H}_t$, $\nu$ by $\frac\nu{2}$ and $\delta$ by $\frac{\delta}{2}$. 
We note that, by~\raf{e3}, the oracle now returns a point $p_{t+1}$ such that
\begin{align}\label{e4}
\xi_t(p_{t+1})\ge\widetilde\xi_t(p_{t+1})\ge\Big(1-\frac\nu2\Big)\max_{q\in H}\widetilde\xi_t(q)\ge\Big(1-\frac\nu2\Big)\cdot\Big(\max_{q\in H}\xi_t(q)-\area(H_t\setminus\widetilde H_t)\Big).
\end{align}
%The following claim states that the total fraction of $\area(H)$ that might not be covered due to this approximation is no more than $\delta/2$. (Recall that $t_f$ denotes the iteration number at which the while-loop terminates.)
\begin{figure}	
	\centering  
	\begin{subfigure}{.4\textwidth}
		\includegraphics[width=2.5in]{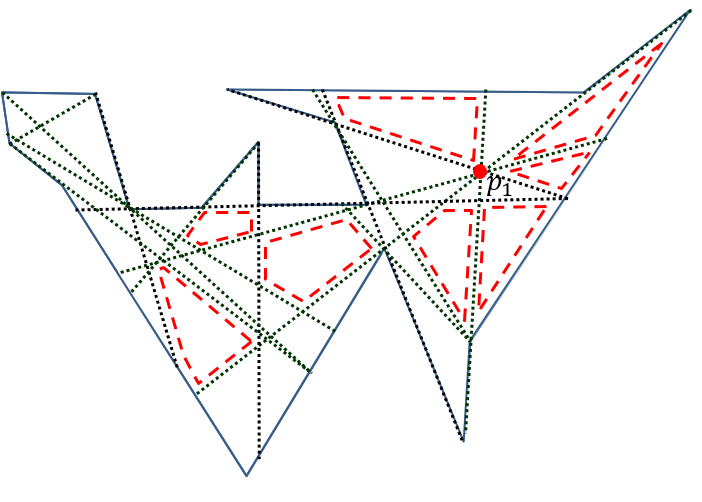}
		\caption{The polygonal arrangement $\cells_1(H)$ and the rounded polygonal arrangement $\widetilde{\cells}_1(H)$. For clarity of presentation, the rounding of only some of the cells in $\cells_1(H)$ is shown (as dashed convex polygons).
		}                                                            \label{fig8}\vspace{.7in}
	\end{subfigure}
	\hspace{0.18in}
	\begin{subfigure}{.4\textwidth}
		\includegraphics[width=2.5in]{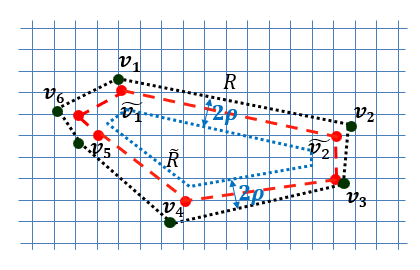}
		\caption{The polygon $R\in\cells_t(H)$ and its internal approximation $\widetilde R$ using the vertices of the grid. Note that, since for all vertices $v$ of $R$ the distance between $v$ and $\widetilde v$ is less than $2\rho$, the difference $R\setminus\widetilde R$ is covered by the region between $\partial R$ and the dotted polygon, which is at distance less than $2\rho$ from $\partial R$.} 
		\label{fig9}
	\end{subfigure}
\caption{Rounding the polygonal arrangement $\cells_t(H)$.}\label{fig89}
\end{figure}
\begin{claim}\label{cl3}                                 
$\area(H_t\setminus\widetilde{H}_t)\le\frac{\nu}{2(1-\nu/2)}\max_{q\in H}\xi_t(q)$.
\end{claim}
\begin{proof}
	Two sets contribute to the difference $H_t\setminus\widetilde{H}_t$: the set of small cells, and the truncated parts of the larges cells $\bigcup_{R\in\cL_t:~P_t(R)<T}R\setminus\widetilde{R}$.
	Note that the total area of the small cells is at most $ r_t\cdot4\rho D$. 
	On the other hand, for any $R\in\cL_t$, we have $\area(R)-\area(\widetilde R)\le 2\rho\cdot\prem(R)$, where $\prem(P)$ is the length of the perimeter of $R$. This inequality holds because $R\setminus\widetilde R$ is contained in the region within distance $2\rho$ from the boundary of $R$; see Figure \ref{fig9} for an illustration. It follows that
	\begin{align*}
	\area(H_t\setminus\widetilde H_t)&=4\rho D r_t +\sum_{R\in\cL_t}\area(R\setminus\widetilde R)\le 2\rho\cdot\sum_{P\in\cL_t}\prem(R)\\&\le 8\rho D r_t
	\le8D r_t\cdot \frac{\nu\delta(1-\varepsilon)^T\area(H)}{16(1-\nu/2) D n r_{\max}}\le\frac{\nu\delta(1-\varepsilon)^T\area(H)}{2(1-\nu/2)  n }\\
	&\le\frac{\nu\cdot w_t(H_t)}{2(1-\nu/2)  n }\le\frac{\nu}{2(1-\nu/2)}\max_{q\in H}\xi_t(q),
	\end{align*}
	by our selection of $\rho$. The claim follows.
\end{proof}
Claim~\ref{cl3}, together with~\raf{e4}, implies that 
\begin{align*}
\xi_t(p_{t+1})\ge(1-\nu)\max_{q\in H}\xi_t(q),
\end{align*}
as required is the proof of Lemma~\ref{l1}.
Note that, since the polygon $H$ is contained in a square of size $D$, the total number of points in $\Gamma$ we need to consider in each principal direction (horizontal and vertical) is at most $$\frac{D}{\rho}\le\frac{32(1-\nu/2) D^2nr_{\max}}{\nu\delta(1-\varepsilon)^T\area(H)}=2^{O(L)}\poly(n,h,\frac{1}{\delta}),$$ 
and thus the number of bits needed to represent each point of $\Gamma$ is $O(L)+\polylog(n,h,\frac{1}{\delta})$. 
Since the vertices of each cell $\widetilde{R}\in\widetilde\cells_t(H)$ lie on the grid, the bit length $L_t$ used in the computations above (in the implementation of the maximization oracle and hence to represent the new point $p_{t+1}$) and the overall running time is $\poly(L,n,h,\log\frac{1}{\delta})$, independent of $t$.
	
%\bibliographystyle{plain}
%\bibliography{ref}	

\section{Alternative Approaches}\label{alt}

%For simplicity, we assume here that $N=G=H$.
\subsection{Standard Greedy Approach}\label{greedy}

The greedy algorithm works by iterating the following until a $(1-\delta)$-fraction of the area of $H$ is seen from some point in $P_t=\{p_1,p_2,\ldots\}$: at time $t$, add a point $p_{t+1}$ (approximately) maximizing the current area of the non-guarded part of $H$. It is easy to see that this gives $(O(\log\frac1\delta),1-\delta)$-approximation.
Indeed, let $p_1,\ldots,p_{t_f}$ be the points chosen by the greedy algorithm until a $(1-\delta)$-fraction of the polygon is guarded. Let $H_t$ be the part of the polygon that is not yet guarded at time $t\in\{0,1,\ldots,t_f\}$. Then at time time $t$, as $H_t$ can be guarded by the optimum solution, $\area(\Vis(p_{t+1}))\ge\frac{\area(H_t)}{\OPT}$, giving $\area(H_{t+1})\le(1-\frac1{\OPT})\area(H_{t})$. Iterating, we get $\area(H_{t_f})\le \delta \area(H)$ for $t_f\ge \OPT\ln\frac1\delta$. Note that that this does not use the fact that the range space $\cF_H$ has bounded VC-dimension.
\subsection{Direct Application of the Br\"{o}nnimann-Goodrich Algorithm }\label{sec:BG}

The original Br\"{o}nnimann-Goodrich algorithm~\cite{BG95}  works for range spaces $(Q,\cR)$ with finite point set $Q$. If one wants to apply this algorithm to the art gallery problem, there are two possible directions:  

\begin{itemize}
	\item Discrtization using a grid $\Gamma$, where the candidate set of guards are restricted to lie on $\Gamma$. The difficulty of this approach is that the size of an optimum set of guards {\it restricted} to $\Gamma$ might be much larger than the size of an optimum set without any restriction. This was the approach used in~\cite{BM16} to give a randomized $O(\log \OPT)$-apporximation algorithm under the general position assumption stated in the introduction.
	
	\item Extension of the Br\"{o}nnimann-Goodrich algorithm for continuous range spaces. It can be shown (see~\cite{E16}) that the algorithm can be extended to work for the  (full coverage version of) the art gallery problem, but the running time may depend on the area of the smallest cell $R$ of the arrangement created in the course of the algorithm. More precisely, the number of iterations is $O(\OPT\cdot\log\frac{\area(Q)}{\area(R)})$. Since $R$ can be arbitrarily small, there is no guarantee that such an extension will terminate in polynomial (ore even finite) time.    
\end{itemize}
\end{document}